\documentclass[journal,onecolumn]{IEEEtran}

\usepackage{graphicx, cite, amsthm, float, dsfont, amsmath, amsfonts, amssymb, enumitem, pgf, tikz, xcolor}

\usepackage{pgfplots}
 
\newtheorem{theorem}{Theorem}

\newtheorem{definition}{Definition}   
      
\newtheorem{remark}{Remark}  
 
\makeatletter   
\newcommand{\vast}{\bBigg@{3}} 
\newcommand{\Vast}{\bBigg@{4}} 
\makeatother 
 
\newcommand{\bs}[1]{\boldsymbol{#1}}
\newcommand{\mcl}[1]{\mathcal{#1}}

\newcommand{\Whe}{\widehat{W}_{\varepsilon}}

\definecolor{mdgreen}{RGB}{0,120,60}
\newcommand{\md}[1]{{\color{black}#1}}
\newcommand{\pooya}[1]{{\color{black}#1}}
\newcommand{\pouya}[1]{{\color{black}#1}}

\begin{document}  
 
\title{The Expurgated Error Exponent is \\Not Universally Achievable}

\author{Seyed AmirPouya Moeini, Marco Dalai and Albert Guill\'en i F\`abregas
\thanks{S. A. Moeini is with the Department of Engineering, University
of Cambridge, CB2 1PZ Cambridge, U.K. (e-mail: sam297@cam.ac.uk). 
Marco Dalai is with the Department of Information Engineering, University
of Brescia, 25123 Brescia, Italy (e-mail: marco.dalai@unibs.it).
A. Guill\'en i F\`abregas is with the Department of Engineering, University
of Cambridge, CB2 1PZ Cambridge, U.K., the Department of
Signal Theory and Communications and the Institute of Mathematics (IMTech), Universitat Polit\`ecnica de Catalunya
08034 Barcelona, Spain  (e-mail: guillen@ieee.org).
}
\thanks{This research was supported in part by the European Research Council under Grant 101142747, the Spanish Government under Grants PID2021-128373OB-I00, PID2025-167766OA-C22, the Catalan Government (AGAUR) under grant 2025 SGR 00902 and the Royal Society International Exchange Grant IES\textbackslash R2\textbackslash 232296.}
\thanks{This work was presented in part at the 2026 International Z\"urich Seminar on Information and Communication, February 2026.}}

 
\maketitle

\begin{abstract}
	We study the universal attainability of the expurgated error exponent for discrete memoryless channels (DMCs). 
	While the random-coding exponent is known to be universally attainable \pooya{for every fixed input distribution}
	 via maximum mutual information (MMI) decoding for DMCs, it remains open whether the expurgated exponent can be attained universally. 
	We show that this is not the case in general. Specifically, we construct a family of DMCs for which no single sequence of codes can attain the expurgated exponent simultaneously for all channels in the family, even at rate zero. 
In addition, for the same channel family, we show that MMI decoding fails to achieve the expurgated exponent for any channel in the family. 
\end{abstract}


\IEEEpeerreviewmaketitle

\section{Introduction}

The random-coding error exponent for discrete memoryless channels (DMC) is the exponential rate of decay of the ensemble average probability of error, averaged over the ensemble of randomly generated codes \cite{gallager, csiszar1977new}. 
It is positive for all rates below capacity \cite{gallager, csiszar1977new}. 
For high rates, the random-coding exponent coincides with the channel reliability function \cite{SHANNON196765}. 
At low rates, the ensemble average error probability is dominated by atypical codes whose codewords are very close to one another.
By removing the worst fraction of codewords from codes in the ensemble \cite{gallager, csiszar1977new, 1056281}, one obtains the so-called expurgated error exponent. The expurgated exponent improves the random-coding exponent at low rates and coincides with the channel reliability function at rate zero \cite{SHANNON196765, SHANNON1967522, 9628087}.

Maximum-likelihood (ML) decoding minimizes the error probability and requires perfect knowledge of the channel. Thus, in order to attaining either the random-coding or expurgated exponents one generally needs to know the channel. 
In \cite[Ch.~10]{csiszarkorner}, Csisz\'ar and K\"orner showed that {for every fixed input distribution} the random-coding exponent is universally attainable for all DMCs. 
Their construction uses the same codebook satisfying the packing lemma \cite[Lem.~10.1]{csiszarkorner} for every DMC, together with maximum mutual information (MMI) decoding, which selects the message whose codeword maximizes the empirical mutual information with the received sequence. The MMI decoding rule does not depend on the channel transition probabilities. 
Since MMI decoding attains the random-coding exponent, this naturally motivates the study of the reliability function under MMI decoding and the question of whether this decoder (or any other universal decoder ignorant of the channel) can achieve higher exponents at lower rates. 
In \cite[p. 10]{1056281}, Csiszár and Körner claimed that the expurgated error exponent is universally attainable for the binary symmetric channel with MMI decoding. 
To the best of our knowledge, whether the expurgated exponent is universally attainable for general DMCs remains open, as noted in \cite[p. 6]{1056281}.
In this direction, Tamir and Merhav \cite{9656768} studied the error exponent of stochastic mutual information (SMI) decoding, the stochastic version of MMI decoding. 
In \cite[Thm. 2]{9656768} it is stated that SMI decoding attains the  expurgated exponent for all DMCs and conjectured that MMI also achieves the expurgated error exponent.
 
In this work, we address the above questions and show that neither holds in general.
We first introduce a 
 family of DMCs for which no single sequence of codes can simultaneously attain the expurgated exponent for all channels in the family, 
thereby showing that this exponent is not universally attainable for all DMCs, in contrast to the random-coding exponent.
We further show that, MMI decoding fails to attain the expurgated exponent for any channel in the same family, seemingly contradicting the generality of Tamir and Merhav's result \cite{9656768} for all DMCs.

These findings do not rule out the possibility that a certain universal decoder may still achieve the expurgated exponent for certain families of channels.
\pouya{
In fact, in Section~\ref{sec:bsc}, we sketch a proof of Csisz\'ar and K\"orner's claim that MMI achieves the expurgated exponent for the binary symmetric channel at all rates.
}
It is therefore of independent interest to characterize the family of channels for which MMI can achieve the expurgated exponent.

The remainder of the paper is organized as follows. Section \ref{sec:prem} introduces the preliminaries and the notation used throughout the paper.
In Section \ref{sec:univ}, we construct a family of DMCs for which no single sequence of codes can attain the expurgated error exponent simultaneously for all channels in the family. 
In Section~\ref{sec:mmi}, we further show that MMI decoding fails to attain the expurgated exponent for any channel in the same family.

\section{Preliminaries}\label{sec:prem}

Scalar random variables are denoted by uppercase letters, their sample values by lowercase letters, and their alphabets by calligraphic letters. Random vectors are represented in boldface.
The type of a sequence $\bs{x}=\left(x_1, \ldots, x_n\right) \in \mcl{X}^n$ is its empirical distribution, defined by
\begin{align}
    \hat{P}_{\bs{x}}(x) \triangleq \frac{1}{n} \sum_{i=1}^n \mathds{1}\!\left\{{x}_i=x\right\}  \,\, \text{ for } x\in \mcl{X}.
\end{align}
The joint type of sequences $\bs{x}=\left(x_1, \ldots, x_n\right) \in \mcl{X}^n,\bs{y}=\left(y_1, \ldots, y_n\right) \in \mcl{Y}^n$, denoted by $\hat{P}_{\bs{x}\bs{y}}(x,y)$ is defined analogously.
For a joint distribution $P$ on $(X,Y)$, we write $I(P)$ to denote the mutual information induced by $P$ and $H_P(X), H_P(X|Y)$ the corresponding marginal and conditional entropies. 

We consider the communication of $M$ equiprobable messages over a DMC $W$,  whose $n$-letter transition probability is
 $W^n(\bs{y}|\bs{x}) = \prod_{i=1}^n W(y_i | x_i)$, for $\bs{x}=(x_1, \ldots, x_n) \in \mcl{X}^n$ and $\bs{y}=(y_1, \ldots, y_n) \in \mcl{Y}^n$,
where $\mcl{X}$ and $\mcl{Y}$ are discrete alphabets with cardinalities $|\mcl{X}|$ and $|\mcl{Y}|$, respectively.
A code of blocklength $n$ is a pair $(f,\phi)$, where the encoder $f$ maps each message to a codeword in $\mcl{X}^n$ and the decoder $\phi$ maps each observed sequence in $\mcl{Y}^n$ to a message.
The rate of such a code is $R=\frac{1}{n}\log M$.

For a given channel $W$ and code $(f,\phi)$, let $p_{e,m}(W,f,\phi)$ denote the error probability when message $m$ is sent, that is, $p_{e,m}(W,f,\phi) = \mathbb{P}\big[\phi(\bs{Y}) \neq m \,|\, m \text{ is sent}\big]$.
The maximal error probability is $p_{{e},\mathrm{max}}(W, f, \phi)\triangleq \max_m p_{{e},m}(W, f, \phi)$,
and the average error probability is
$p_{{e}}(W, f, \phi)\triangleq\frac{1}{M}\sum_{m=1}^M p_{{e},m}(W, f, \phi)$.
For a given codebook $\mcl{C}=\left\{\bs{x}_1,\ldots,\bs{x}_M\right\}$, the ML decoder is optimal in the sense that it minimizes the average error probability.
Given the channel output $\bs{y}$, it selects the message
\begin{align}
	\widehat{m} = \underset{m\in \{1,\ldots,M\}}{\operatorname{argmax}} W^n(\bs{y}|\bs{x}_m),
\end{align}
and therefore requires perfect knowledge of the channel transition probabilities.
The MMI decoder selects the message $\widehat{m}$ that maximizes the mutual information induced by the joint type of $\bs{y}$ and each codeword,
\begin{align}
\widehat{m} = \underset{m\in \{1,\ldots,M\}}{\operatorname{argmax}}\, I(\hat{P}_{\bs{x}_m\bs{y}}).
\label{eq:mmi_dec}
\end{align}
Observe that the computation of the mutual information in \eqref{eq:mmi_dec} does not involve the knowledge of the channel statistics.

\pooya{For a given rate $R$ and a channel $W$, an error exponent $E$ is said to be \emph{achievable} if there exists a sequence of codes $(f_n,\phi_n)$ with $M_n$ messages such that}
\begin{align} \label{eq-achievability-defn}
	&\liminf_{n\rightarrow \infty} \frac{1}{n}\log M_n \geq R,\notag\\
    &
	\liminf_{n \rightarrow \infty} -\frac{1}{n}\log p_{{e}}(W, f_n, \phi_n) \geq E.
\end{align}

\pouya{A decoding rule $\phi_n$ is said to attain error exponent $E$ if there exists a sequence of encoders $f_n$ for which the conditions in \eqref{eq-achievability-defn} hold under that decoding rule $\phi_n$.}
Achievability can similarly be stated in terms of the maximal error probability rather than the average error probability. 
Two well-known achievable error exponents are the random-coding exponent and the expurgated exponent (see \cite[Ch.~5]{gallager}).
The reliability function of a channel $W$ is defined as the supremum of all achievable error exponents at rate $R$.
A complete characterization of this function is still unknown,  \pooya{but coincides with the random-coding exponent at high rates \cite{SHANNON196765} and with the expurgated exponent in the limit as $R\downarrow0$
\cite{SHANNON1967522,9687534,9628087}.}

\md{For a fixed input distribution $Q$ on $\mathcal{X}$, the \pooya{constant-composition} random-coding exponent \cite{csiszarkorner} is
\begin{align}
E_{\mathrm{r}}(R,Q,W)
	\triangleq \min_{V}\,
	D(V\Vert W\mid Q)+\left[I(Q,V)-R\right]^+,
	\label{eq:random-coding-ck}
\end{align}
where $[t]^+\triangleq \max\{0,t\}$ and $D(V\Vert W| Q)\triangleq\sum_xQ(x)D(V(\cdot|x)\Vert W(\cdot|x))$. 
The \pooya{constant-composition} expurgated exponent \pooya{for a fixed input distribution $Q$} is
\begin{align}
	&E_{\mathrm{ex}}(R,Q,W)
	\notag\\
    &\quad\triangleq
	\min_{\substack{P_{X\bar{X}}:\,P_X=P_{\bar{X}}=Q\\
	I_P(X;\bar{X})\leq R}}
	\,\mathbb{E}_{P}\!\left[d_W(X,\bar{X})\right]
	+I_P(X;\bar{X})-R.
	\label{eq:expurgated-ck}
\end{align}
where $d_W(x,\bar{x})\triangleq-\log\sum_y\sqrt{W(y|x)W(y|\bar{x})}$ \pooya{is the Bhattacharyya distance}.
In particular,
\begin{align}\label{eq:rate-zero-expurgated-exponent}
	\lim_{R\downarrow0}E_{\mathrm{ex}}(R,Q,W)
	=\sum_{(x,\bar{x})}Q(x)Q(\bar{x})d_W(x,\bar{x}).
\end{align}
When the input distribution is optimized, we write 
$E_{\mathrm{ex}}(R,W)\triangleq\max_Q E_{\mathrm{ex}}(R,Q,W)$.}
\pooya{
Observe that the constant-composition expressions in \eqref{eq:random-coding-ck} and \eqref{eq:expurgated-ck} are given in the primal domain, 
whereas the corresponding i.i.d. random-coding and expurgated exponents were derived by Gallager in the dual domain \cite[Ch.~5]{gallager}.
For each exponent, the constant-composition expression is at least as large as its i.i.d. counterpart for every fixed $Q$, but the two coincide after optimization over $Q$ \cite{csiszarkorner}.
}

We now introduce the definition of a universally attainable error exponent as in \cite[Def.~10.7]{csiszarkorner}.

\begin{definition}\label{defn: universal}
	A function $E^*(W)$ of $W$ \pouya{(denoted channel function)} is a \emph{universally attainable error exponent} at rate $R>0$ for the family of DMCs $\{W:\mathcal{X}\to\mathcal{Y}\}$ if, for every $\delta>0$
	 there exists an integer $n_0(|\mathcal{X}|,|\mathcal{Y}|,R,\delta)$ such that for all
	$n\geq n_0$, there exist $n$-length block codes $(f_{\pooya{n}},\phi_{\pooya{n}})$ of rate at least $R-\delta$ such that
	\begin{align}
		p_{e,\max}(W,f_n,\phi_n)
		\leq \exp\!\left[-n\big(E^*(W)-\delta\big)\right]
	\end{align}
	for every DMC $\{W:\mathcal{X}\to\mathcal{Y}\}$.
\end{definition}

{
This is a strong notion of universality, as it requires the same sequence of encoder--decoder pairs $(f_n,\phi_n)$ to attain the exponent for every channel in the family.
If an input distribution $Q$ is prescribed, we write the channel function as $W \mapsto E(R,Q,W)$.
The code may depend on $R,Q,\delta$, and $n$, but not on $W$.
Other definitions of universality have been studied in the literature; see, for example, \cite{705540} and the references therein.
}


\pooya{The random-coding exponent is known to satisfy Definition \ref{defn: universal}.}
\begin{theorem}[\!\!{\cite[Thm.~10.8]{csiszarkorner}}]\label{thm:random-coding}
	For every fixed distribution $Q$ on $\mathcal{X}$, the random-coding exponent $E_{\mathrm{r}}(R,Q,W)$ is universally attainable at rate $R$.
	\label{thm:ck-random-universal} 
\end{theorem}

\pouya{
Observe that for fixed $R$ and $Q$, a single sequence of encoders
whose codewords satisfy the packing lemma
\cite[Lem.~10.1]{csiszarkorner} achieves both the random-coding 
and expurgated exponents for every DMC $W$ under ML decoding \cite{1056281}.
Since the packing conditions do not depend on the channel,
the encoder can be chosen independently of $W$, while the
ML decoder still requires knowledge of $W$.}
Theorem \ref{thm:random-coding} leads to the main question of this paper, \pooya{also raised in \cite[p.~6]{1056281}}: for general $R$ and $Q$, is the channel function $W\mapsto E_{\mathrm{ex}}(R,Q,W)$ universally attainable? 
We show that the answer is negative. 
Our counterexample consists of two distinct channels whose expurgated exponents coincide under the equiprobable input distribution. 
We show that, for every sufficiently small rate $R$, any code that attains this exponent on one channel fails to attain it for the other.


\section{A Counterexample}
\label{sec:univ}


\md{In this section, we introduce two DMCs for which, for every sufficiently small positive rate and for the same fixed input distribution, no single sequence of codes can attain the expurgated exponent on both channels. 
This proves that the channel function $W\mapsto E_{\mathrm{ex}}(R,Q,W)$ is not universally attainable. 
\pooya{
We first present the argument in the rate-zero limit, where it is easiest to see, and then extend it to positive rates.
}}
 
For $0 \le \varepsilon \le 1$, define the binary-input quaternary-output channels with alphabets $\mcl{X}=\{0,1\}$ and $\mcl{Y}=\{a,b,c,d\}$ as
\begin{align}\label{eq:weps}
	&W_{\varepsilon}=
		\frac12
\begin{bmatrix}
1-\varepsilon&1-\varepsilon&\varepsilon&\varepsilon\\
\varepsilon&\varepsilon&1-\varepsilon&1-\varepsilon
	\end{bmatrix} ,
\end{align}
and 
\begin{align}
	&\widehat{W}_{\varepsilon}=\frac12
\begin{bmatrix}
1-\varepsilon&\varepsilon&1-\varepsilon&\varepsilon\\
\varepsilon&1-\varepsilon&\varepsilon&1-\varepsilon
\end{bmatrix}.
\label{eq:wheps}
\end{align}
Observe that $\widehat{W}_{\varepsilon}$ is obtained from $W_{\varepsilon}$ by permuting the output symbols $b$ and $c$.
From the perspective of ML decoding, both channels are equivalent to a binary symmetric channel with crossover probability $\varepsilon$.
\md{The two channels have the same Bhattacharyya distances, and hence}
\begin{align}
	\md{E_{\mathrm{ex}}(R,Q,W_{\varepsilon})
	=E_{\mathrm{ex}}(R,Q,\widehat{W}_{\varepsilon})}
	\label{eq:equal-expurgated}
\end{align}
\md{for every input distribution $Q$ and every rate $R$. 
Let $Q_{\mathrm{u}}$ denote the equiprobable distribution on $\{0,1\}$, \pooya{which maximizes the expurgated exponent for both channels at every rate.}
\md{By \eqref{eq:rate-zero-expurgated-exponent}, the common rate-zero limit is}
\begin{align}\label{eq:ml-exp-rate-zero}
	\md{\lim_{R \downarrow 0} \,E_{\mathrm{ex}}(R,Q_{\mathrm{u}}, W) = -\frac{1}{2}\log  2 \sqrt{\varepsilon(1-\varepsilon)}.}
\end{align}}
\md{Choose $\varepsilon\in(0,1/2)$ such that\footnote{$\varepsilon < 0.0149353113291091\ldots$ is sufficient.}}
\begin{align}\label{eq:cond-varepsilon}
	\md{-\log \left(\frac{1-\varepsilon}{2}\right)<-\frac{1}{2}\log  2 \sqrt{\varepsilon(1-\varepsilon)}.}
\end{align}
\pooya{
Set ${E}_{\varepsilon}\triangleq -\log \left(\frac{1-\varepsilon}{2}\right)$. 
}

We next define the mapping $\kappa:\{0,1\}^2 \to \{a,b,c,d\}$ by  
\begin{align}
	&\kappa(0, 0) = a, \quad \kappa(0, 1) = b, \notag\\
    &\kappa(1, 0) = c, \quad \kappa(1,1) = d.
	\label{eq:kappa}
\end{align}
For any pair of sequences $(\bs{x},\bar{\bs{x}})\in\mcl{X}^n$, let $\bs{y}_{\kappa}(\bs{x},\bar{\bs{x}})$ be the sequence defined by $y_i\triangleq\kappa(x_i,\bar{x}_i)$ for $i=1,\ldots,n$. 
Since $\kappa$ is a bijection, each symbol $y_i$ uniquely specifies the pair $(x_i,\bar{x}_i)$, and thus $(\bs{x},\bar{\bs{x}})$ is determined componentwise from $\bs{y}_{\kappa}(\bs{x},\bar{\bs{x}})$.
Moreover, by construction, $y_i\in\{a,b\}$ when $x_i=0$ and $y_i\in\{c,d\}$ when $x_i=1$, and hence $W_{\varepsilon}(y_i|x_i)=(1-\varepsilon)/2$ for all $i$. 
Likewise, $y_i\in\{a,c\}$ when $\bar{x}_i=0$ and $y_i\in\{b,d\}$ when $\bar{x}_i=1$, and hence $\Whe(y_i|\bar{x}_i)=(1-\varepsilon)/2$ for all $i$. 
Therefore, this construction ensures that, for every such pair $(\bs{x},\bar{\bs{x}})$,
\begin{align}\label{eq:wwepsrel}
	&W_{\varepsilon}^n\big(\bs{y}_{\kappa}(\bs{x},\bar{\bs{x}}) \big| \bs{x}\big)\notag\\
    &\quad=\widehat{W}_{\varepsilon}^n\big(\bs{y}_{\kappa}(\bs{x},\bar{\bs{x}}) \big| \bs{\bar{x}}\big)= \left(\frac{1-\varepsilon}{2} \right)^n.
\end{align}

\pooya{Now consider any block code $(f_n,\phi_n)$ with at least
two messages, and choose distinct messages $m_1$ and $m_2$. Let
$\bs{y}^*=\bs{y}_{\kappa}(f_n(m_1),f_n(m_2))$. 
If $\phi_n(\bs{y}^*)\neq m_1$, receiving $\bs{y}^*$ results in a decoding error on $W_{\varepsilon}$ when $m_1$ is sent. Hence,
\begin{align}
	p_{e,\max}(W_{\varepsilon},f_n,\phi_n)
	&\geq
	\sum_{\bs{y}:\,\phi_n(\bs{y})\neq m_1}
	W_{\varepsilon}^n(\bs{y}| f_n(m_1))\\
	&\geq W_{\varepsilon}^n(\bs{y}^*| f_n(m_1))\\
	&=\left(\frac{1-\varepsilon}{2} \right)^n \label{eq-18-eq} \\
	&=\exp(-n{E}_{\varepsilon}),
\end{align}
where \eqref{eq-18-eq} follows from \eqref{eq:wwepsrel}.
If instead $\phi_n(\bs{y}^*)=m_1$, receiving $\bs{y}^*$ causes
an error on $\widehat W_{\varepsilon}$ when $m_2$ is sent.
Similarly,
\begin{align}
	p_{e,\max}(\widehat W_{\varepsilon},f_n,\phi_n)
	&\geq
	\sum_{\bs{y}:\,\phi_n(\bs{y})\neq m_2}
	\widehat W_{\varepsilon}^n(\bs{y}| f_n(m_2))\\
	&\geq \widehat W_{\varepsilon}^n(\bs{y}^*| f_n(m_2))\\
	&=\left(\frac{1-\varepsilon}{2} \right)^n  \\
	&=\exp(-n{E}_{\varepsilon}).
\end{align}
Thus, for every pair of distinct messages $m_1$ and $m_2$,
\begin{align}
	&\max\!\left\{
	p_{e,m_1}(W_{\varepsilon},f_n,\phi_n),
	p_{e,m_2}(\widehat W_{\varepsilon},f_n,\phi_n)
	\right\}
	\notag\\
    &\qquad\geq \exp(-n{E}_{\varepsilon}).
	\label{eq:universal-dichotomy-1}
\end{align}
In particular, this implies that
\begin{align}
	&\max\!\left\{
	p_{e,\max}(W_{\varepsilon},f_n,\phi_n),
	p_{e,\max}(\widehat W_{\varepsilon},f_n,\phi_n)
	\right\}
	\notag\\
    &\qquad\geq \exp(-n{E}_{\varepsilon}).
	\label{eq:universal-dichotomy}
\end{align}
In other words, for any code, this single output sequence forces the maximal error probability to be at least $\exp(-n{E}_{\varepsilon})$ on at least one of the two channels.
Consequently, whenever both expurgated exponents exceed ${E}_{\varepsilon}$, a sequence of codes attaining the expurgated exponent on one channel cannot attain it for the other.

Although Definition~\ref{defn: universal} concerns positive rates, the rate-zero meaning of \eqref{eq:universal-dichotomy} is precise. Consider any nontrivial zero-rate sequence, 
by which we mean $M_n\geq2$ and $\frac{1}{n}\log M_n\downarrow0$. 
By \eqref{eq:cond-varepsilon}, the common expurgated exponent in \eqref{eq:ml-exp-rate-zero} is strictly larger than ${E}_{\varepsilon}$, 
whereas \eqref{eq:universal-dichotomy} holds for every $n$. 
Hence no such sequence can attain the rate-zero expurgated exponents simultaneously on both channels. This also covers the usual stronger convention in which $M_n\to\infty$.

Fig. \ref{fig:funcs} below compares the converse exponent $-\log\big((1-\varepsilon)/2\big)$ with the rate-zero expurgated exponent $-\tfrac{1}{2}\log\big(2\sqrt{\varepsilon(1-\varepsilon)}\big)$ 
and the rate-zero random-coding exponent $-\log\!\big(\tfrac{1}{2}+\sqrt{\varepsilon(1-\varepsilon)}\big)$ for these channels. 
Observe that when $\varepsilon=0$, the converse exponent coincides with the rate-zero random-coding exponent, while the rate-zero expurgated exponent diverges.}

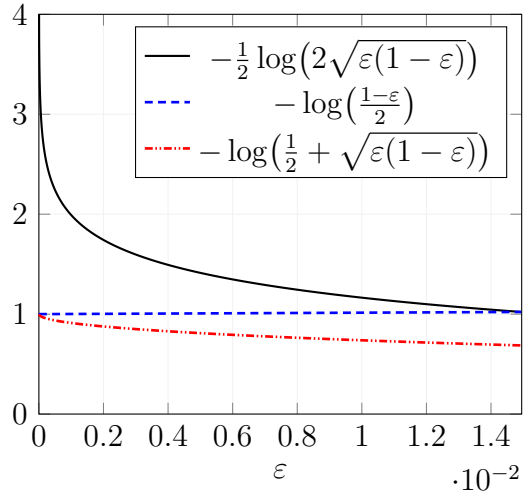
\begin{figure}[H]
 \centering
\begin{tikzpicture}
\begin{axis}[
  width=8cm,
  xlabel={$\varepsilon$},
%
  xmin=0, xmax=0.01495, 
  ymin=0, ymax=4,
%
  grid=major,
  minor tick num=0,
  xtick={0,0.002,0.004,0.006,0.008,0.010,0.012,0.014},
  ytick={0,1,2,3,4},
  major grid style={line width=0.3pt, draw=gray!10},
%
  legend pos=north east,
  legend style={ 
    draw=black,
    fill=white, 
    nodes={inner sep=2.5pt},
  },
]
\addplot [thick, black, line width=0.8pt] table[x=x,y=y] {\jobname-ml-expurgated.dat};
\addlegendentry{$-\frac{1}{2}\log\!\big(2\sqrt{\varepsilon(1-\varepsilon)}\big)$}

\addplot [semithick, blue, densely dashed, line width=1.0pt] table[x=x,y=y] {\jobname-converse.dat};
\addlegendentry{$-\log\!\left(\tfrac{1-\varepsilon}{2}\right)$}

\addplot [semithick, red, densely dashdotdotted,, line width=1.0pt] table[x=x,y=y] {\jobname-random-coding.dat};
\addlegendentry{$-\log\!\big(\tfrac{1}{2}+\sqrt{\varepsilon(1-\varepsilon)}\big)$}

\end{axis}
\end{tikzpicture}
\caption{The converse exponent $-\log\!\left(\tfrac{1-\varepsilon}{2}\right)$,
	the rate-zero expurgated exponent $-\frac{1}{2}\log\!\big(2\sqrt{\varepsilon(1-\varepsilon)}\big)$,
	and the rate-zero random-coding exponent $-\log\!\big(\tfrac{1}{2}+\sqrt{\varepsilon(1-\varepsilon)}\big)$
	for the channels $W_{\varepsilon}$ and $\widehat{W}_{\varepsilon}$. Rates and error exponents in this figure are expressed using base-2 logarithms.
		}
    \label{fig:funcs}
\end{figure}

\pooya{
The converse bound \eqref{eq:universal-dichotomy} does not depend on the rate. 
To turn the rate-zero argument into the positive-rate statement of Definition~\ref{defn: universal}, specialize \eqref{eq:expurgated-ck} to $Q_{\mathrm{u}}$. For either channel,
\begin{align}\label{eq-exp-exp}
	&E_{\mathrm{ex}}(R,Q_{\mathrm{u}},W)
	\notag\\
    &\;= \min_{\substack{0\leq\alpha\leq1\\
	\log 2-h_2(\alpha)\leq R}}
	\Bigl\{
	-\alpha\log\!\left(2\sqrt{\varepsilon(1-\varepsilon)}\right)
	\notag\\
    &\hspace{7em}{}+\log 2-h_2(\alpha)-R
	\Bigr\},
\end{align}
where $h_2(\alpha)\triangleq-\alpha\log\alpha-(1-\alpha)\log(1-\alpha)$. This expression is continuous from the right at $R=0$. Therefore, \eqref{eq:cond-varepsilon} implies that there exists $R_{\varepsilon}>0$ such that
\begin{align}
	{E}_{\varepsilon}<E_{\mathrm{ex}}(R,Q_{\mathrm{u}},W)
	\label{eq:positive-rate-gap}
\end{align}
for both channels and every $R\in(0,R_{\varepsilon})$. We take $R_{\varepsilon}$ to be the largest such threshold.
Fix $R\in(0,R_{\varepsilon})$ and choose
\begin{align}
	0<\delta<\min\!\left\{R,
	E_{\mathrm{ex}}(R,Q_{\mathrm{u}},W_{\varepsilon})-{E}_{\varepsilon}\right\}.
	\label{eq:delta-gap}
\end{align}
This choice ensures that $R-\delta>0$ and that both expurgated exponents remain strictly above ${E}_{\varepsilon}$.
Suppose that a single sequence of codes attains the expurgated exponent on both channels at rate $R$.
For all sufficiently large $n$, these codes have rate at least $R-\delta>0$ and hence at least two messages.
Every such code must therefore still satisfy \eqref{eq:universal-dichotomy}: at least one of the two maximal error probabilities must be at least $\exp(-n{E}_{\varepsilon})$.
But our choice $R\in(0,R_{\varepsilon})$ ensures that both expurgated exponents are strictly larger than ${E}_{\varepsilon}$.
With $\delta$ chosen as in \eqref{eq:delta-gap}, attaining these exponents would therefore require both maximal error probabilities to be strictly smaller than $\exp(-n{E}_{\varepsilon})$ for all sufficiently large $n$.
This contradicts the bound that every such code must satisfy, and establishes the following theorem.

\begin{theorem}\label{thm:universal}
	There exist a fixed input distribution $Q$, two DMCs
	$\mathcal{W}=\{W_{\varepsilon},\widehat W_{\varepsilon}\}$, and a threshold $R_{\varepsilon}>0$ such that, for every $R\in(0,R_{\varepsilon})$, the channel function
	$W\mapsto E_{\mathrm{ex}}(R,Q,W)$ is not universally attainable  over $\mathcal{W}$ at rate $R$. 
\end{theorem}}

\pooya{
\begin{remark}
	The two channels have the same expurgated exponent for every input distribution. 
	The equiprobable distribution is optimal for both channels, but this choice is not essential to the argument, since the lower bound in \eqref{eq:universal-dichotomy} does not depend on the input distribution. 
	For example, consider
	\begin{align}
		Q_\alpha=(\alpha,1-\alpha), \qquad \varepsilon=0.001.
	\end{align}
	For both channels, the rate-zero expurgated exponent exceeds the converse exponent when
	\begin{align}
		-2\alpha(1-\alpha)\log\!\left(2\sqrt{\varepsilon(1-\varepsilon)}\right) > -\log\!\left(\frac{1-\varepsilon}{2}\right),
	\end{align}
	which gives, approximately, $ 0.14743 <\alpha<0.85257$. This interval includes $\alpha=1/2$, corresponding to $Q_{\mathrm u}$.
	For each fixed distribution in this interval, the same strict inequality holds at sufficiently small positive rates by continuity. 
	Thus, Theorem \ref{thm:universal} also holds for each of these distributions, with a rate threshold that may depend on the distribution.
\end{remark}}


\pooya{
\begin{remark}
Theorem~\ref{thm:universal} also holds when the maximal error probability in Definition~\ref{defn: universal} is replaced by the average error probability. To see this, consider any code $(f_n,\phi_n)$ with $M_n\geq2$ messages. If every message has error probability at least $\exp(-nE_{\varepsilon})$ on $W_{\varepsilon}$, then the average error probability satisfies the same bound. 
Otherwise, there must exist a message $m_0$ such that
\begin{align} \label{eq:m0choose}
p_{e,m_0}(W_{\varepsilon},f_n,\phi_n)
<\exp(-nE_{\varepsilon}).
\end{align}
For every other message $m$, \eqref{eq:universal-dichotomy-1} implies that at least one of $p_{e,m_0}(W_{\varepsilon},f_n,\phi_n)$ and $p_{e,m}(\widehat W_{\varepsilon},f_n,\phi_n)$ must be at least $\exp(-nE_{\varepsilon})$. 
Since \eqref{eq:m0choose} rules out the first possibility, we have
\begin{align}
p_{e,m}(\widehat W_{\varepsilon},f_n,\phi_n)
\geq\exp(-nE_{\varepsilon}).
\end{align}
Thus, on at least one channel, at least $M_n-1$ messages satisfy the lower bound.
Averaging over all $M_n$ messages gives
\begin{align}
&\max\!\left\{
p_e(W_{\varepsilon},f_n,\phi_n),
p_e(\widehat W_{\varepsilon},f_n,\phi_n)
\right\}\notag\\
&\quad\geq \frac{M_n-1}{M_n}\exp(-nE_{\varepsilon}) \\
&\quad\geq \frac12\exp(-nE_{\varepsilon}).
\label{eq:universal-dichotomy-average}
\end{align}
The factor $1/2$ does not affect the error exponent, so the same conclusions hold for every $R\in(0,R_{\varepsilon})$ and for nontrivial zero-rate sequences.
\end{remark}}

\section{A Counterexample for MMI Decoding}
\label{sec:mmi}

\pooya{
The previous section shows that no single sequence of codes can attain the expurgated exponent on both channels simultaneously at sufficiently small positive rates. 
In this section, we prove a stronger result for MMI decoding: at these rates, the expurgated exponent cannot be attained on either channel individually, even when the \pouya{encoders} are chosen specifically for that channel.
}
As a consequence, MMI cannot attain the expurgated exponent for all DMCs, seemingly contradicting Tamir and Merhav \cite{9656768}.
Nevertheless, MMI may still achieve the expurgated exponent for certain families of channels.
\pouya{
In Section~\ref{sec:bsc}, we sketch a proof that MMI decoding achieves the expurgated exponent for the binary symmetric channel at all rates, as stated by Csisz\'ar and K\"orner in \cite{1056281}.
}

Let the channel be $W_{\varepsilon}$ defined in \eqref{eq:weps}, with $\varepsilon$ satisfying \eqref{eq:cond-varepsilon}.
For $M\geq 3$, consider any constant-composition collection $\mathcal{C}_n=\{\bs{x}_1,\ldots,\bs{x}_M\}$ of blocklength $n$, and let $f_n(m)=\bs{x}_m$ and $\phi_n$ denote the encoder and the MMI decoder, respectively.
Fix a message $m$, and choose $\bar m$ such that $\bs{x}_m$ and $\bs{x}_{\bar m}$ are not bitwise complements. Such a $\bar m$ always exists when $M \ge 3$, since a binary sequence has at most one complement.
The error probability when sending message $m$ is at least the probability that the decoder \pooya{strictly favours} $\bar{m}$ over $m$, that is,
\begin{align}\label{eq:mbarmerr}
    p_{{e},m}(W_{\varepsilon}, f_n, \phi_n) &\geq \mathbb{P}\!\left[I\big(\hat{P}_{\bs{x}_{m}\bs{Y}}\big) < I\big(\hat{P}_{\bs{x}_{\bar{m}}\bs{Y}}\big)\right].
\end{align}
Note that in \eqref{eq:mbarmerr}, ties are not counted as errors.
Consider the sequence $\bs{y}_{\kappa}(\bs{x}_m,\bs{x}_{\bar{m}})$ defined in the previous Section. We will construct a sequence $\tilde{\bs{y}}$ from $\bs{y}_{\kappa}(\bs{x}_m,\bs{x}_{\bar{m}})$
for which the strict inequality in \eqref{eq:mbarmerr} holds.
First observe that $I\big(\hat{P}_{\bs{x}_{m}\bs{y}_{\kappa}(\bs{x}_m,\bs{x}_{\bar{m}})}\big)=I\big(\hat{P}_{\bs{x}_{\bar{m}}\bs{y}_{\kappa}(\bs{x}_m,\bs{x}_{\bar{m}})}\big)$. 
Specifically,
\begin{align}
    I\big(\hat{P}_{\bs{x}_{m}\bs{y}_{\kappa}(\bs{x}_m,\bs{x}_{\bar{m}})}\big)
    &=H_{\hat{P}_{\bs{x}_m}}(X) - H_{\hat{P}_{\bs{x}_{m}\bs{y}_{\kappa}(\bs{x}_m,\bs{x}_{\bar{m}})}}(X | Y) \label{eq:twomatch1} \\
    &=H_{\hat{P}_{\bs{x}_m}}\!(X)\\
    I\big(\hat{P}_{\bs{x}_{\bar{m}}\bs{y}_{\kappa}(\bs{x}_m,\bs{x}_{\bar{m}})}\big)
    &=H_{\hat{P}_{\bs{x}_{\bar{m}}}}(\bar{X}) - H_{\hat{P}_{\bs{x}_{\bar{m}}\bs{y}_{\kappa}(\bs{x}_m,\bs{x}_{\bar{m}})}}(\bar{X} | Y)\\
    &=H_{\hat{P}_{\bs{x}_{\bar{m}}}}(\bar{X}). \label{eq:twomatch2}
\end{align}
By the construction of the bijective mapping $\kappa$ in \eqref{eq:kappa}, the output sequence $\bs{y}_{\kappa}(\bs{x}_m,\bs{x}_{\bar{m}})$ determines both $\bs{x}_m$ and $\bs{x}_{\bar{m}}$ symbol by symbol.
Hence, both conditional entropies are zero.
Since the code is constant-composition, the two marginal entropies are equal, so the mutual informations match.


We now modify the output sequence $\bs{y}_{\kappa}(\bs{x}_m,\bs{x}_{\bar m})$ in a single position. 
Since $\bs{x}_m$ and $\bs{x}_{\bar m}$ have the same type but are not the complement of each other, we have $\hat{P}_{\bs{x}_m\bs{x}_{\bar{m}}}(0,1)=\hat{P}_{\bs{x}_m\bs{x}_{\bar{m}}}(1,0)>0$, 
and at least one of $\hat{P}_{\bs{x}_m\bs{x}_{\bar{m}}}(0,0)$ or $\hat{P}_{\bs{x}_m\bs{x}_{\bar{m}}}(1,1)$ is positive.
If $\hat{P}_{\bs{x}_m\bs{x}_{\bar{m}}}(0,0)>0$, then select an index $i$ such that $(x_{m,i},x_{\bar m,i})=(1,0)$, and modify $\bs{y}_{\kappa}(\bs{x}_m,\bs{x}_{\bar{m}})$ by changing only its $i$th symbol from $c$ to $a$.
Otherwise, if $\hat{P}_{\bs{x}_m\bs{x}_{\bar{m}}}(1,1)>0$, select an index $i$ with $(x_{m,i},x_{\bar m,i})=(0,1)$, and change only the $i$th symbol from $b$ to $d$.
Denote the modified sequence by $\tilde{\bs{y}}$.
By construction, $\tilde{\bs{y}}$ still determines $\bs{x}_{\bar m}$ symbol by symbol since both $c$ and $a$ correspond to $\bar{x}=0$, and both $b$ and $d$ correspond to $\bar{x}=1$.
Hence $H_{\hat P_{\bs{x}_{\bar m}\tilde{\bs{y}}}}(\bar{X}|Y)=0$, and therefore
\begin{align}
	I\big(\hat P_{\bs{x}_{\bar m}\tilde{\bs{y}}}\big)=H_{\hat P_{\bs{x}_{\bar m}}}(\bar{X}).
\end{align}
However, the modification breaks the property that $X$ is determined by $Y$ from the joint types of $(\bs{x}_m,\tilde{\bs{y}})$.
Indeed, after changing $c$ to $a$ (or $b$ to $d$), there exists a symbol $y\in\{a,d\}$ that appears in $\tilde{\bs{y}}$ both at some position where $x_{m,i}=0$ and at another position where $x_{m,j}=1$. 
Thus, $\hat{P}_{\bs{x}_m\tilde{\bs{y}}}(0,y)>0$ and $\hat{P}_{\bs{x}_m\tilde{\bs{y}}}(1,y)>0$, and hence $H_{\hat{P}_{\bs{x}_m\tilde{\bs{y}}}}(X|Y)>0$. Therefore,
\begin{align} 
	I\big(\hat P_{\bs{x}_m\tilde{\bs{y}}}\big) 
	&= H_{\hat P_{\bs{x}_m}}(X) - H_{\hat P_{\bs{x}_m\tilde{\bs{y}}}}(X|Y) \\ 
	&< H_{\hat P_{\bs{x}_m}}(X)\\ &=H_{\hat P_{{\bs{x}_{\bar{m}}}}}(\bar{X}). 
\end{align} 
Consequently, we conclude that
\begin{align} 
	I\big(\hat P_{\bs{x}_m\tilde{\bs{y}}}\big)< I\big(\hat P_{{\bs{x}_{\bar{m}}}\tilde {\bs{y}}}\big) . 
\end{align}
Hence, for the given pair $(m,\bar m)$, we have constructed an output sequence under which the MMI decoder strictly favours $\bar m$ over $m$ when $m$ was sent.
In particular, the error probability is at least the probability of observing $\tilde{\bs{y}}$.
Since $\tilde{\bs{y}}$ differs from $\bs{y}_{\kappa}(\bs{x}_m,\bs{x}_{\bar m})$ in exactly one position, the output probability at that coordinate changes from $(1-\varepsilon)/2$ to $\varepsilon/2$. 
Thus,
\begin{align}
	W_\varepsilon^n(\tilde{\bs{y}}|\bs{x}_m) &= \Big(\frac{1-\varepsilon}{2}\Big)^{n-1}\frac{\varepsilon}{2}\\
	&= \Big(\frac{1-\varepsilon}{2}\Big)^{n}\frac{\varepsilon}{1-\varepsilon}.
\end{align}
Combining the above gives
\begin{align}
	p_{e,m}(W_\varepsilon,f_n,\phi_n) \ge \Big(\frac{1-\varepsilon}{2}\Big)^{n}\frac{\varepsilon}{1-\varepsilon}.
\end{align}
Observe that this lower bound does not depend on the particular message $m$, so
\begin{align}\label{eq:mmi-converse}
	-\frac{1}{n}\log p_e(W_\varepsilon,f_n,\phi_n) \leq -\log\Big(\frac{1-\varepsilon}{2}\Big) + o(1).
\end{align}
Moreover, the argument does not depend on the specific codebook and applies to any constant-composition code.
Therefore, \eqref{eq:mmi-converse} holds for every such code under MMI decoding, including the optimal one.
From the previous section, for every rate $R$ in $[0,{R}_{\varepsilon})$, this converse exponent lies strictly below the expurgated exponent of the channel at that rate.
This completes the counterexample, showing that no constant-composition code can attain the expurgated exponent for this channel under MMI decoding.
If $n$ is large enough, one can extract from any block code a constant-composition subcode by taking its largest type class, with negligible loss in rate and without increasing the error probability \cite[p. 95]{SHANNON196765}. 
Therefore, the above argument in fact holds for any code of sufficiently large blocklength.

We summarize the conclusion of this section in the following theorem.

\begin{theorem}
\label{thm:mmi}
	There exists a family of DMCs $\mathcal{W}$ and a rate threshold $R_{\mathcal{W}}>0$ with the following property. 
	For every channel $W\in\mathcal{W}$ and every rate $R\in[0,R_{\mathcal{W}})$, the reliability function of the channel under MMI decoding is strictly smaller than the expurgated error exponent of $W$ at rate $R$.
\end{theorem}


Fig. \ref{fig:mmi_fig} compares the expurgated error exponent with the converse exponent for $\varepsilon = 0.001$.
For this value of $\varepsilon$, which satisfies \eqref{eq:cond-varepsilon}, the two exponents differ over a nontrivial range of rates.

\begin{figure}[H]
 \centering
\begin{tikzpicture}
\begin{axis}[
  width=8cm,
  xlabel={$R$ (bits/use)}, 
%
  xmin=0, xmax=0.205,
  ymin=0.7, ymax=2.2,
  grid=major,
  minor tick num=0,
  major grid style={line width=0.3pt, draw=gray!10},
  legend pos=north east,
  legend style={ 
    draw=black,
    fill=white,
    nodes={inner sep=2.5pt},
  },
]
\addplot [thick, black, line width=0.8pt] table[x=x,y=y] {\jobname-mmi-case.dat};
\addlegendentry{$E_{\mathrm{ex}}(R,W_{\varepsilon})$}

\addplot [semithick, black, densely dotted, line width=1.0pt] table[x=x,y=y] {\jobname-mmi-converse.dat};
\addlegendentry{${E}_{\varepsilon}$}

\end{axis}
\end{tikzpicture}
    \caption{Comparison between the expurgated error exponent and the converse exponent 
        ${E}_{\varepsilon}=-\log\!\big[(1-\varepsilon)/2\big]$ for $\varepsilon = 0.001$.  
        In this example, ${R}_{\varepsilon} \approx 0.1865$ bits/use, so any code of rate below this value cannot attain the expurgated exponent under MMI decoding.}
        \label{fig:mmi_fig}
\end{figure}
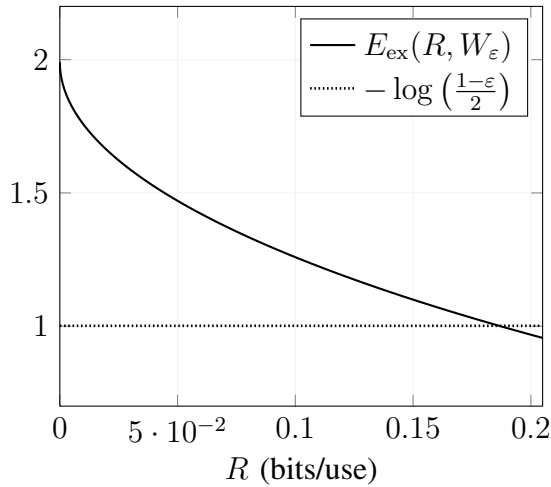

So far, we have considered deterministic decoders. We now show that stochastic mutual information (SMI) decoding also fails to attain the expurgated exponent for the channels considered here.
From the result below, whose proof can be found in the Appendix, it follows that for any nonnegative decoding metric $q$, the error exponent of the maximum metric decoder is always at least as good as that of the corresponding stochastic decoder. 

\begin{theorem}\label{thm:comp}
	Fix an encoder $f$ and a deterministic maximum-metric decoder $\phi$ based on a nonnegative metric $q$. Let $\hat{\phi}$ be the associated stochastic decoder that, 
	given $\bs{y}$, selects $\hat m$ with probability proportional to 	$q\big(f(m),\bs{y}\big)$. Then, for every message $m$ and every DMC $W$,
	\begin{align}
		p_{e,m}(W,f,\phi)\le 2\,p_{e,m}(W,f,\hat{\phi}).
	\end{align}
	
\end{theorem}

The above result and Theorem \ref{thm:mmi} imply that the SMI decoder also fails to attain the expurgated error exponent for all DMCs, thus providing a counterexample to Tamir and Merhav's result for all DMCs \cite{9656768}.
We emphasize that this counterexample does not rule out the possibility that MMI or SMI may attain the expurgated error exponent for certain families of channels. 
\pouya{Indeed, in the next section, we sketch a proof that MMI decoding achieves the expurgated exponent for binary symmetric channels at all rates, as stated in \cite{1056281}.}

\begingroup
\section{MMI Decoding for Binary Symmetric Channels}
\label{sec:bsc}

We sketch a proof that MMI achieves the expurgated exponent for binary
symmetric channels (BSCs) at all rates, as stated in \cite{1056281}.
Let $W$ be a BSC with crossover probability $0<\varepsilon<1/2$.
For simplicity, assume that $n$ is even and consider codewords containing
$n/2$ ones.
For a fixed received sequence $\bs{y}$, the empirical mutual information
can be shown to increase with $|n/2-d_{\text{h}}(\bs{x}_m,\bs{y})|$,
where $d_{\text{h}}$ denotes Hamming distance.
Thus, MMI decoding is equivalent to
\begin{align}
    \widehat m \in \underset{m}{\operatorname{argmin}}\,
    \min\!\left\{
        d_{\text{h}}(\bs{x}_m,\bs{y}),\,
        n-d_{\text{h}}(\bs{x}_m,\bs{y})
    \right\}.
    \label{eq:bsc-mmi-distance}
\end{align}
In other words, the decoder finds the closest sequence among the
codewords and their bitwise complements, where each codeword and its
complement are assigned to the same message.

An error when message $m$ is sent therefore requires some competing
codeword or its complement to be at least as close to $\bs{y}$ as
$\bs{x}_m$.
For the BSC, the pairwise Bhattacharyya bound for two sequences at
distance $d$ is $b_\varepsilon^d$, where
$b_\varepsilon\triangleq 2\sqrt{\varepsilon(1-\varepsilon)}$.
Applying this bound to each competitor and its complement gives
\begin{align}
    &p_{e,m}(W,f_n,\phi_n)
    \leq \sum_{j\neq m}
    \left(b_\varepsilon^{d_{mj}}
          +b_\varepsilon^{n-d_{mj}}\right),
    \notag\\
    &\qquad
    d_{mj}\triangleq d_{\text{h}}(\bs{x}_m,\bs{x}_j).
    \label{eq:bsc-pairwise-bound}
\end{align}
This counts ties as errors and therefore holds for any tie-breaking rule.

We now choose the codewords so that this bound gives the expurgated
exponent.
Fix $0<r<R<\log 2$ and draw $2M$ words independently and equiprobably
from the type class with $n/2$ ones, where $M=\lceil e^{nr}\rceil$.
Let $N_m(d)$ denote the number of other codewords at distance $d$ from
$\bs{x}_m$.
A word in this type class at distance $d$ from $\bs{x}_m$ differs
in exactly $d/2$ of its zero positions and $d/2$ of its one positions.
Hence, averaging over the randomly drawn codewords gives, for even $d$,
\begin{align}
    \mathbb{E}[N_m(d)]
    &=(2M-1)
      \frac{\dbinom{\frac{n}{2}}{\frac{d}{2}}^{2}}
           {\dbinom{n}{\frac{n}{2}}}
      \\
    &\leq (2M-1)(n+1)
      e^{-nC\left(\frac{d}{n}\right)},
    \notag\\
    &\qquad C(t)\triangleq\log 2-h_2(t),
    \label{eq:bsc-mean-spectrum}
\end{align}
where $h_2(t)$ denotes the binary entropy function.
To obtain a bound for every message, we apply expurgation
\cite[Ch.~5]{gallager}.
For each $m$, Markov's inequality and a union bound over the
possible distances give
\begin{align}
    &\mathbb{P}\!\Bigl[
        N_m(d)>2(n+1)\mathbb{E}[N_m(d)]
        \text{ for some }d\in\{0,2,\ldots,n\}
    \Bigr]
    \nonumber\\
    &\qquad\leq
    \sum_{d=0,2,\ldots,n}\frac{1}{2(n+1)}
    \leq \frac12.
    \label{eq:bsc-expurgation-selection}
\end{align}
The expected number of codewords satisfying
$N_m(d)\leq 2(n+1)\mathbb{E}[N_m(d)]$ for every $d$ is therefore
at least $M$.
Choose a codebook containing at least $M$ such codewords and keep
only $M$ of these.
Removing codewords can only decrease the counts, so the resulting
code satisfies
\begin{align}
    N_m(d)
    &\leq 2(2M-1)(n+1)^2e^{-nC(d/n)}
    \\
    &\leq e^{n[R-C(d/n)]}
    \label{eq:bsc-distance-spectrum}
\end{align}
for every message and every distance.
The last inequality holds for all sufficiently large $n$ because
$M=\lceil e^{nr}\rceil$ and $r<R$.
This selection does not depend on $\varepsilon$.
Also, if $C(d/n)>R$, the bound is smaller than one, forcing
$N_m(d)=0$.
Thus, only distances satisfying $C(d/n)\leq R$ occur in the
resulting code.

By \eqref{eq:bsc-pairwise-bound} and
\eqref{eq:bsc-distance-spectrum}, each such distance contributes
at most
$e^{n[R-C(d/n)]}(b_\varepsilon^d+b_\varepsilon^{n-d})$
to the error bound.
Using $C(t)=C(1-t)$ and summing over at most $n+1$ distances,
we obtain
\begin{align}
    &-\frac1n\log p_{e,\max}(W,f_n,\phi_n)
    \notag\\
    &\quad\geq
    \min_{\substack{0\leq t\leq 1/2\\C(t)\leq R}}
    \left\{-t\log b_\varepsilon+C(t)-R\right\}-o(1)
    \\
    &\quad=E_{\mathrm{ex}}(R,W)-o(1).
    \label{eq:bsc-mmi-expurgated}
\end{align}
The last equality follows from \eqref{eq:expurgated-ck} with the
equiprobable input distribution, which is optimal for the BSC.
The $o(1)$ term can be chosen independently of $\varepsilon$,
and $r<R$ can be chosen arbitrarily close to $R$.
Since both the encoder and the MMI decoder are independent of
$\varepsilon$, this proves universal attainability in the sense
of Definition~\ref{defn: universal}.
The same bounds hold for odd $n$ using codewords with
$\lfloor n/2\rfloor$ ones.
The case $\varepsilon>1/2$ follows by complementing the output.
\par
\endgroup

\begin{appendices}

\section{Proof of Theorem \ref{thm:comp}}
\label{app:a}


	Let $\mcl{C}=\{\bs{x}_1,\ldots,\bs{x}_M\}$ be the codebook, so that $f(m)=\bs{x}_m$. Then
	\begin{align}
		&p_{e,m}(W,f,\phi)
		\notag\\
    &\;\leq \mathbb{P}\!\left[\bigcup_{m'\neq m} q(\bs{x}_m,\bs{Y}) \leq q(\bs{x}_{m'},\bs{Y}) \right]\label{eq:to-comp-0} \\
		&= \sum_{\bs{y}} W^n(\bs{y}|\bs{x}_m) \, \mathds{1}\!\left\{\bigcup_{m'\neq m}q(\bs{x}_m,\bs{y}) \leq q(\bs{x}_{m'},\bs{y}) \right\}\\
		&\leq \sum_{\bs{y}} W^n(\bs{y}|\bs{x}_m) \, \min\left\{1, \sum_{m'\neq m} \frac{q(\bs{x}_{m'},\bs{y})}{q(\bs{x}_m,\bs{y}) }\right\}. \label{eq:to-comp-1}
	\end{align}
	The upper bound in \eqref{eq:to-comp-0} follows since we count ties as errors. 
	For the stochastic decoder,
	\begin{align}
		&p_{e,m}(W,f,\hat{\phi})\notag\\
    &\;= 1 - \sum_{\bs{y}} W^n(\bs{y}|\bs{x}_m) \frac{q(\bs{x}_m,\bs{y})}{\sum_{m'} q(\bs{x}_{m'},\bs{y})}\\
		&\;= 1 - \sum_{\bs{y}} W^n(\bs{y}|\bs{x}_m) \frac{1}{1 + \sum_{m\neq m'} \frac{q(\bs{x}_{m'},\bs{y})}{q(\bs{x}_m,\bs{y})}}.
	\end{align}
	Using $\frac{1}{1+\beta}\le 1-\frac12\min\{1,\beta\}$ for $\beta\geq 0$, we obtain
	\begin{align}
		&p_{e,m}(W,f,\hat{\phi})\notag\\
    &\;\geq 1 - \sum_{\bs{y}} W^n(\bs{y}|\bs{x}_m)\notag\\
    &\hspace{4em}\biggl(1- \frac{1}{2}\min\left\{1,\sum_{m'\neq m}\frac{q(\bs{x}_{m'},\bs{y})}{q(\bs{x}_m,\bs{y})} \right\} \biggr)\\
		&\;= \frac{1}{2}\sum_{\bs{y}} W^n(\bs{y}|\bs{x}_m) \min\left\{1,\sum_{m'\neq m}\frac{q(\bs{x}_{m'},\bs{y})}{q(\bs{x}_m,\bs{y})} \right\} . \label{eq:to-comp-2}
	\end{align}
	Comparing \eqref{eq:to-comp-1} and \eqref{eq:to-comp-2} yields the claim.

\end{appendices}

\begingroup
\pouya{
\section*{Acknowledgment}
We thank the anonymous reviewer for suggesting a way to extend our rate-zero result for MMI decoding over binary symmetric channels to all rates. This motivated the content of Section \ref	{sec:bsc}.
}
\par
\endgroup

\ifCLASSOPTIONcaptionsoff
  \newpage
\fi

{\footnotesize
\bibliographystyle{IEEEtran}
\bibliography{refs}
}

\begin{IEEEbiographynophoto}{Seyed AmirPouya Moeini}
received the bachelor's degrees in computer engineering and mathematics
from Sharif University of Technology, Tehran, Iran, in 2023.
He is currently pursuing the Ph.D.\ degree with the Department of
Engineering, University of Cambridge, U.K. He is also a member of
Trinity Hall. His research interests include information theory, applied
probability, and combinatorics.
\end{IEEEbiographynophoto}

\begin{IEEEbiographynophoto}{Marco Dalai}
(Senior Member, IEEE) received the degree in electronic engineering and the
Ph.D.\ degree in information engineering from the University of Brescia,
Italy, in 2003 and 2007, respectively. He is currently an Associate Professor
with the Department of Information Engineering, University of Brescia. He is
a member of the IEEE Information Theory Society and was a member of the
Board of Governors from 2022 to 2024. He was the recipient of the 2014 IEEE
Information Theory Society Paper Award. He has been an Associate Editor of
the IEEE \textsc{Transactions on Information Theory} since March 2020.
\end{IEEEbiographynophoto}

\begin{IEEEbiographynophoto}{Albert Guill\'en i F\`abregas}
(Fellow, IEEE) received the degree in telecommunications engineering from the
Universitat Polit\`ecnica de Catalunya (UPC), the degree in electronics engineering
from the Politecnico di Torino in 1999, and the Ph.D.\ degree in communication
systems from \'Ecole Polytechnique F\'ed\'erale de Lausanne (EPFL) in 2004.

He is a Professor at the University of Cambridge, a fellow of the
Emmanuel College, and a Distinguished Researcher (part-time) at
UPC. He has held appointments with the New Jersey Institute of Technology, Telecom
Italia, European Space Agency (ESA), ICREA, Institut Eur\'ecom, the
University of South Australia, and Universitat Pompeu Fabra; as well as
visiting appointments with EPFL, T\'el\'ecom Paris, the University of South Australia, the
Centrum Wiskunde \& Informatica, and Texas A\&M University at Qatar.  His research interests include the
areas of information theory, communication theory, coding theory, and
statistical inference.

Dr.\ Guill\'en i F\`abregas is a fellow of the Institute for Mathematics and
its Applications (IMA) and a member of the Young Academy of Europe. He
received the Starting, Consolidator, Advanced, and Proof of Concept Grants
from European Research Council, the 2004 European
Signal Processing Conference Young Authors Award, the 2004 Best Doctoral Thesis Award from Spanish
Institution of Telecommunications Engineers, and the Research Fellowship from
Spanish Government to join ESA. Since 2013, he has been an Editor of
\emph{Foundations and Trends in Communications and Information Theory}
(Emerald Publishing). He was an Associate Editor of IEEE \textsc{Transactions
on Wireless Communications} from 2007 to 2011 and IEEE \textsc{Transactions on
Information Theory} from 2013 to 2020.
\end{IEEEbiographynophoto}
\vfill 

\end{document}